\documentclass[12pt]{article}

\usepackage{times}
\usepackage{fullpage}
\usepackage{amsfonts,amssymb,amsmath,amsthm}
\usepackage{latexsym} 
\usepackage{gauss}
\usepackage{tikz}
\usepackage{url}
\usepackage{thmtools}
\usepackage{thm-restate}
\usepackage{hyperref}
\usepackage[capitalise, nameinlink]{cleveref}
\usepackage{algorithm}
\usepackage{algorithmicx}
\usepackage{algpseudocode}
\usepackage[shortlabels]{enumitem}

\newcommand{\R}{\mathbb{R}}

\newcommand{\Tr}{\mathrm{Tr}}

\newcommand{\ADV}{\mathrm{ADV}}

\def\01{\{0,1\}}

\newcommand{\ket}[1]{|#1\rangle}
\newcommand{\braket}[2]{\langle#1 | #2\rangle}
\newcommand{\ketbra}[2]{| #1 \rangle \langle #2 |}

\newtheorem{theorem}{Theorem}

\newtheorem{lemma}[theorem]{Lemma}
\newtheorem{corollary}[theorem]{Corollary}

\newtheorem{remark}[theorem]{Remark}

\newtheorem{claim}[theorem]{Claim}
\newtheorem{fact}[theorem]{Fact}

\theoremstyle{definition}

\newtheorem{definition}[theorem]{Definition}

\begin{document}
\title{The quantum query complexity of composition with a relation}
\author{Aleksandrs Belovs \thanks{Faculty of Computing, University of Latvia. Email: stiboh@gmail.com} \and 
Troy Lee \thanks{Centre for Quantum Software and Information, University of Technology Sydney.  Email: troyjlee@gmail.com}}
\date{}
\maketitle

\begin{abstract}
The negative weight adversary method, $\ADV^\pm(g)$, is known to characterize the bounded-error quantum query 
complexity of any Boolean function $g$, and also obeys a perfect composition theorem $\ADV^\pm(f \circ g^n) = \ADV^\pm(f) \ADV^\pm(g)$. 
Belovs \cite{Bel15} gave a modified version of the negative weight adversary method, $\ADV_{rel}^\pm(f)$, that characterizes the 
bounded-error quantum query complexity of a relation $f \subseteq \{0,1\}^n \times [K]$, provided the relation is efficiently verifiable.  
A relation is efficiently verifiable if $\ADV^\pm(f_a) = o(\ADV_{rel}^\pm(f))$ for every $a \in [K]$, where $f_a$ is the Boolean function defined as 
$f_a(x) = 1$ if and only if $(x,a) \in f$.  In this note we show a perfect composition theorem for the composition of a relation $f$ with a Boolean function $g$
\[
\ADV_{rel}^\pm(f \circ g^n)  = \ADV_{rel}^\pm(f) \ADV^\pm(g) \enspace .
\]
For an efficiently verifiable relation $f$ this means $Q(f \circ g^n) = \Theta( \ADV_{rel}^\pm(f) \ADV^\pm(g) )$.
\end{abstract}

\section{Introduction}
Quantum query complexity has been a very successful model for studying quantum algorithms.  
The most famous quantum algorithms, like Grover's search algorithm \cite{Gro96} and the period finding routine of 
Shor's algorithm \cite{Sho97}, can be formulated in this model, and the model has also been fruitful 
for developing new algorithmic techniques like quantum walks \cite{Amb07, Sze04, MSS07} and learning graphs 
\cite{Bel12}.  

One of the greatest successes of quantum query complexity is the adversary method.  The (unweighted) adversary 
method began as a lower bound technique developed by Ambainis \cite{Amb02} to show a lower bound on the 
quantum query complexity of the two-level AND-OR tree, among other problems.  It was later generalized to a weighted version in various 
forms by several authors \cite{Amb06, Zha05, BSS03, LM08}.  \v{S}palek and Szegedy \cite{SS06} showed that all of these 
generalizations were in fact equivalent.

A strictly stronger bound called the negative weights adversary method, $\ADV^\pm(f)$, was developed by H{\o}yer et al. \cite{HLS07}.
In a wonderful turn of events, a series of works by Reichardt et al. \cite{RS12, Rei14, Rei11} showed that the bounded-error quantum query 
complexity, $Q(f)$, of a Boolean function $f$ satisfies $Q(f) = O(\ADV^\pm(f))$, and therefore $\ADV^\pm(f)$ characterizes bounded-error quantum query complexity up 
to a constant multiplicative factor.  

This characterization has several interesting consequences.  For one, it means that upper bounds on quantum 
query complexity on a function $f$ can now be shown by upper bounding $\ADV^\pm(f)$, which can be expressed 
as a relatively simple semidefinite program.  
This approach has led to improved algorithms for many problems of interest \cite{Bel12b,Bel15b,LMS17}, especially via the development of the 
learning graphs model of Belovs \cite{Bel12}.

Another consequence is that quantum query complexity inherits the nice properties of the adversary bound.  
One of these nice properties is that the adversary bound behaves perfectly with respect to \emph{function composition}.
For Boolean functions $f:\{0,1\}^n \rightarrow \{0,1\}$ and $g: \{0,1\}^m \rightarrow \{0,1\}$ define the 
composition $h = f \circ g^n$ to be the function $h: \{0,1\}^{nm} \rightarrow \{0,1\}$ where for an input 
$x=(x_1, \ldots, x_n) \in \{0,1\}^{nm}$, with each $x_i \in \{0,1\}^m$, we have $h(x) = f(g(x_1), \ldots, g(x_n))$.
\cite{HLS07} showed that $\ADV^\pm(h) \ge \ADV^\pm(f) \ADV^\pm(g)$, and \cite{Rei14} showed a matching 
upper bound.  

\begin{theorem}[\cite{HLS07, Rei14}]
\label{thm:basic}
Let $f:\{0,1\}^n \rightarrow \{0,1\}$ and $g: \{0,1\}^m \rightarrow \{0,1\}$.  
\[
\ADV^\pm(f \circ g^n) = \ADV^\pm(f) \ADV^\pm(g) \enspace .
\]
\end{theorem}

This perfect composition theorem was later extended by \cite{LMRSS11} to allow $g$ 
to be a partial Boolean function and allow the range of $f$ to be non-Boolean.  Kimmel \cite{Kim13} 
showed a perfect composition theorem where both $f$ and $g$ were allowed to be partial Boolean functions.  
Thus here the domain of $h = f \circ g^n$ is only those $x=(x_1, \ldots, x_n)$ where $(g(x_1), \ldots, g(x_n))$ is in the domain of $f$.

Extensions of the negative weight adversary method have been given for the more general query complexity problem of 
\emph{state generation} \cite{Shi02, AMRR11,LMRSS11}.  In this problem, on input $x$ the algorithm begins in the state $\ket{0}\ket{0}$ as usual 
but now the goal is to transform this state into the state $\ket{\sigma_x} \ket{0}$ for some target vector $\ket{\sigma_x}$ by making queries to $x$.  Lee 
et al. \cite{LMRSS11} have shown that an extension of the negative weight adversary method called the filtered $\gamma_2$ norm gives a \emph{semi-tight} characterization of 
the quantum query complexity of the state conversion problem.  It is semi-tight because a slightly larger error is needed in the upper bound 
than in the lower bound.  

Belovs \cite{Bel15} used this characterization to give a modified adversary bound, $\ADV_{rel}^\pm(f)$, that characterizes the quantum query complexity 
of a relation $f \subseteq \{0,1\}^n \times [K]$, provided the relation is \emph{efficiently verifiable}.  Intuitively, a relation is efficiently verifiable if given $a$ the 
complexity of checking if $(x,a)$ is in the relation is low-order compared to $\ADV_{rel}^\pm(f)$; the formal definition 
is given in \cref{def:verifiable}.  For an efficiently verifiable relation the success probability of an algorithm can be amplified without increasing the order of the complexity,
getting around the semi-tightness of the \cite{LMRSS11} characterization.

In this work, we show a perfect composition theorem $\ADV_{rel}^\pm(f \circ g^n) = \ADV_{rel}^\pm(f) \ADV^\pm(g)$ for the composition of a relation 
$f \subseteq \{0,1\}^n \times [K]$ with Boolean function $g$.  If $f$ is efficiently verifiable this implies that $Q(f \circ g^n) = \Theta( \ADV_{rel}^\pm(f) \ADV^\pm(g) )$. 
The lower-bound part of this theorem was required to show a lower bound on the runtime of a quantum algorithm constructing a cut sparsifier of a graph \cite{AW19}.
The perfect adversary composition theorem for functions has been a very useful tool, both for constructing algorithms and showing lower bounds, and we believe our composition 
theorem for relations will find additional applications as well.

\section{Preliminaries}
For a positive integer $m$ we let $[m] = \{1, \ldots, m\}$.  For two matrices $A,B$ of the same size, $A \circ B$ denotes 
the Hadamard or entrywise product: $(A \circ B)(x,y) = A(x,y) B(x,y)$.   We use $\|A\|$ to denote the 
spectral norm of a matrix $A$ and $\|A\|_{tr}$ to denote the trace norm.  For a symmetric matrix $A$, we use $\lambda_{max}(A)$ and  
$\lambda_{min}(A)$
to denote the largest and smallest eigenvalues of $A$, respectively.  We use $I_n$ for the $n$-by-$n$ identity matrix, and drop the 
subscript when the dimension is implied from context.
For a vector $v \in \R^n$ and natural numbers 
$a \le b$ we let $v(a:b) \in \R^{b-a+1}$ be the vector $(v(a), v(a+1), \ldots, v(b))$.  

We will need a few simple facts about positive semidefinite matrices.
First, if $A, B\succeq 0$, then $A\otimes B\succeq 0$.  Also, a principal submatrix of a positive semidefinite matrix is positive semidefinite.
A matrix $\tilde A$ obtained from $A$ by duplicating rows and columns is again positive semidefinite.
The last two facts follow from the following more general observations:

\begin{fact}
\label{fact:duplication}
Let $A\succeq 0$ be an $M$-by-$M$ matrix, and $h \colon [N] \to [M]$ a function.
Then the $N$-by-$N$ matrix $\tilde A$ defined by $\tilde A(x,y) = A(h(x), h(y))$ is also positive semidefinite.
\end{fact}

\begin{fact}
\label{fact:simple}
If $A,B$ are the same size and $A \succeq 0, B \succeq 0$ 
then $A \circ B \succeq 0$.  Similarly, if $A \succeq 0, B \preceq 0$ then $A \circ B \preceq 0$.
\end{fact}

\begin{fact}
\label{fact:lmax}
Let $A$ be a symmetric matrix.  Then $\lambda_{max}(A) \cdot I \succeq A \succeq \lambda_{min}(A) \cdot I$.  In particular, 
$\lambda_{max}(A) \ge v^T A v \ge \lambda_{min}(A)$ for any unit vector $v$.  
\end{fact}

\subsection{Quantum query complexity}
The bounded-error quantum query complexity of the function $g: \{0,1\}^n \rightarrow \{0,1\}$, denoted $Q(g)$, is the
minimum number of queries needed by a quantum algorithm that outputs $g(x)$ with probability at least $2/3$
for every input $x \in \{0,1\}^n$.  Let $K$ be a positive integer and $f \subseteq \{0,1\}^n \times [K]$ be a relation.
We say that a quantum algorithm computes $f$ if for every $x \in \{0,1\}^n$ the algorithm outputs an $a \in [K]$ such that 
$(x,a) \in f$ with probability at least $2/3$.  We let $Q(f)$ denote the minimum cost of a quantum query algorithm that computes $f$.

We will always assume that for every $x \in \{0,1\}^n$ there exists $a \in [K]$ such that $(x,a) \in f$.
This assumption is without loss of generality as if $f$ defined on a strict subset $S$ of $\{0,1\}^n$ we can consider instead 
$f'$ which is defined as $f$ together with every pair $(x,a) \in (\{0,1\}^n\setminus S) \times [K]$.
In this way any algorithm that computes $f$ on $S$ also computes $f'$ on $\{0,1\}^n$, as any output is accepted for inputs not in $S$.

\subsection{Adversary bound}
\begin{definition}[Adversary matrix for a function]
Let $g: \{0,1\}^n \rightarrow \{0,1\}$ be a Boolean function.  Let $\Gamma_g \in \R^{2^n \times 2^n}$, with rows and columns labeled by 
elements of $\{0,1\}^n$.  We say that $\Gamma_g$ is a functional adversary matrix for $g$ 
iff $\Gamma_g$ is symmetric and $\Gamma_g(x,y) = 0$ for all $x,y$ such that $g(x) = g(y)$.
\end{definition}

Up to a permutation of rows and columns, we may always assume that the first $|g^{-1}(0)|$ rows and columns of a functional adversary 
matrix $\Gamma_g$ are 
labeled by elements of $g^{-1}(0)$.  Then there is a $|g^{-1}(0)|$-by-$|g^{-1}(1)|$ matrix $Z$ such that 
\[
\Gamma_g = 
\begin{bmatrix}
0 & Z \\
Z^T & 0
\end{bmatrix}
\enspace.
\]
In this paper we always assume that adversary matrices for functions are presented in this form.

\begin{definition}[Functional adversary bound \cite{HLS07}]
\label{def:adversary}
Let $g: \01^n \rightarrow \{0,1\}$ be a Boolean function.  For $i \in [n]$, let $D_i$ be a $2^n$-by-$2^n$ Boolean matrix 
where $D_i(x,y) = 1$ iff $x_i \ne y_i$.
The adversary bound $\ADV^{\pm}(g)$ for $g$ is defined as 
\begin{equation*}
\begin{aligned}
\ADV^\pm(g) = & \underset{\Gamma}{\text{maximize}}
& & \|\Gamma\| \\
& \text{subject to}
& & \|\Gamma \circ D_i \| \leq 1 \mbox{ for all } i = 1, \ldots, n \\
& & & \Gamma(x,y) = 0 \mbox{ for all } x,y \in \{0,1\}^n \mbox{ with } g(x) = g(y) \enspace .
\end{aligned}
\end{equation*}
\end{definition}

We will also need the dual formulation of the adversary bound.  
\begin{theorem}[\cite{LMRSS11} Theorem 3.4]
\label{thm:adv_dual}
Let $g: \01^n \rightarrow \{0,1\}$ be a Boolean function.  Let $\{u_{x,i}\}, \{v_{x,i}\}$ be two families of 
vectors of arbitrary finite dimension indexed by $x \in \{0,1\}^n$ and $i \in [n]$.
\begin{equation*}
\begin{aligned}
\ADV^\pm(g) = & \underset{\{u_{x,i}\},  \{v_{x,i}\} }{\text{minimize}}
& & \max \left\{ \max_{x \in \{0,1\}^n} \sum_{i \in [n]} \| u_{x,i}\|^2, \max_{x \in \{0,1\}^n} \sum_{i \in [n]} \| v_{x,i}\|^2 \right\} \\
& \text{subject to}
& & \sum_{i: x_i \ne y_i} \langle u_{x,i}, v_{y, i} \rangle = 
\begin{cases}
1 & \mbox{ if } g(x) \ne g(y) \\
0 & \mbox{ if } g(x) = g(y)
\end{cases}
\mbox{ for all } x,y \in \{0,1\}^n \\
\end{aligned}
\end{equation*}
\end{theorem}
If one simply takes the dual of \cref{def:adversary} one gets the above optimization problem without the constraint 
on $x,y$ pairs where $g(x) = g(y)$.  These additional constraints are needed to show the upper bound in the composition 
theorem, and \cite{LMRSS11} show they can be added without increasing the objective value of the program.

The functional adversary bound characterizes the bounded-error quantum 
query complexity of any function $g$.  The lower bound is due to H{\o}yer et al. \cite{HLS07} and 
the upper bound due to Reichardt \cite{Rei11}.
\begin{theorem}[\cite{HLS07},\cite{Rei11}]
\label{thm:adv_char}
Let $g: \01^n \rightarrow \{0,1\}$ be a Boolean function.  Then $Q(g) = \Theta(\ADV^{\pm}(g))$
\end{theorem}

\section{Adversary bound for relations}
Belovs \cite{Bel15} developed a modification, $\ADV_{rel}^\pm(f)$, of the adversary bound that relates to the bounded-error quantum query complexity 
of a relation $f \subseteq \{0,1\}^n \times [K]$.  To motivate this bound, we first review the \emph{state generation} problem \cite{Shi02, AMRR11,LMRSS11}.  
A state generation problem is specified by a family of vectors $\ket{\sigma_x} \in \R^M$ for each $x \in \{0,1\}^n$.  On input $x$ the algorithm begins in the state $\ket{0}\ket{0}$ 
and the goal is for the algorithm to finish in the target state $\ket{\sigma_x}\ket{0}$ after making as few queries to $x$ as possible.

Lee et al. \cite{LMRSS11} give the following definition and theorem.
\begin{definition}[Filtered $\gamma_2$ norm]
Let $A$ be a $2^n$-by-$2^n$ matrix and $D = \{D_1, \ldots, D_n\}$ where each 
$D_i$ is a $2^n$-by-$2^n$ Boolean matrix defined as $D_i(x,y) = 1$ iff $x_i \ne y_i$.  Define
\begin{equation*}
\begin{aligned}
\gamma_2(A | D) = & \underset{\{u_{x,i}\},  \{v_{x,i}\} }{\text{minimize}}
& & \max \left\{ \max_{x \in \{0,1\}^n} \sum_{i \in [n]} \| u_{x,i}\|^2, \max_{x \in \{0,1\}^n} \sum_{i \in [n]} \| v_{x,i}\|^2 \right\} \\
& \text{subject to}
& & \sum_{i \in [n]} \langle u_{x,i}, v_{y, i} \rangle \cdot D_i(x,y) = A(x,y) \mbox{ for all } x,y \in \{0,1\}^n \\
\end{aligned}
\end{equation*}
\end{definition}

\begin{theorem}[\cite{LMRSS11}]
\label{thm:lmrss}
Let $M,n$ be positive integers and $\{\ket{\sigma_x}\}_{x \in \{0,1\}^n}$ be a family of vectors with each $\ket{\sigma_x} \in \R^M$.  Let $A$ 
be a $2^n$-by-$2^n$ matrix where $A(x,y) = 1- \braket{\sigma_x}{\sigma_y}$ for all $x,y \in \{0,1\}^n$.  Let $D = \{D_1, \ldots, D_n\}$ where each 
$D_i$ is a $2^n$-by-$2^n$ Boolean matrix defined as $D_i(x,y) = 1$ iff $x_i \ne y_i$.
For any $0 < \epsilon < \gamma_2(A | D)$ there is a quantum algorithm that for every $x \in \{0,1\}^n$ terminates in a state $\ket{\sigma_x'}$ satisfying 
$\langle  \sigma_x' | (|\sigma_x \rangle \otimes | 0 \rangle) \ge \sqrt{1-\epsilon}$ after making $O( \gamma_2(A | D) \frac{\log(1/\epsilon)}{\epsilon^2})$ 
many queries to $x$.
\end{theorem}

In computing a relation $f \subseteq \{0,1\}^n \times [K]$, there is not a fixed ideal target state; rather the algorithm has the freedom to optimize over a set of 
target states which we call perfect target states for $f$.  
\begin{definition}[Perfect target states]
Let $f \subseteq \{0,1\}^n \times [K]$ and $\{\ket{\psi_x}\}_{x \in \{0,1\}^n}$ a family of unit vectors.  We say that $\{\ket{\psi_x}\}_{x \in \{0,1\}^n}$ are \emph{perfect 
target states} for $f$ if there exists a complete family of orthogonal projectors $\{\Pi_a\}_{a \in [K]}$ such that each $\ket{\psi_x}$ can be 
decomposed as $\ket{\psi_x} = \sum_{a: (x,a) \in f} \ket{\sigma_{x,a}}$ for some (un-normalized) vectors $\ket{\sigma_{x,a}}$ satisfying $\Pi_a \ket{\sigma_{x,a}} = \ket{\sigma_{x,a}}$ 
for all $x,a$.
\end{definition}

Let $\{ \ket{\psi_x}\}_{x \in \01^n}$ be a family of perfect target states for $f$ and suppose there is a state generation algorithm that on input $x$ terminates in the 
state $\ket{\psi_x}\ket{0}$ for every $x \in \{0,1\}^n$.
If the algorithm measures according to the projectors $\{\Pi_a \otimes I\}_{a \in [K]}$, for the projectors $\{\Pi_a\}_{a \in [K]}$ witnessing the perfection of 
$\{ \ket{\psi_x}\}_{x \in \01^n}$, then for every $x$ it will output an $a$ with $(x,a) \in f$ with certainty.  This motivates studying the optimization problem of minimizing the filtered 
$\gamma_2$ norm $\gamma_2( 1 - [\braket{\psi_x}{\psi_y}]_{x,y \in \{0,1\}^n} | D)$ over all families of perfect target states $\{\ket{\psi_x}\}_{x \in \01^n}$ for $f$.   
This optimization problem gives Belovs' definition of the relational adversary bound.

\begin{definition}[\cite{Bel15} Equation~20]
\label{def:adv_rel}
Let $f \subseteq \{0,1\}^n \times [K]$.
\begin{equation*}
\begin{aligned}
\ADV_{rel}^\pm(f) = & \underset{\{u_{x,i}\},  \{v_{x,i}\}, \{\sigma_{x,a}\} }{\text{minimize}}
& & \max \left\{ \max_{x \in \{0,1\}^n} \sum_{i \in [n]} \| u_{x,i}\|^2, \max_{x \in \{0,1\}^n} \sum_{i \in [n]} \| v_{x,i}\|^2 \right\} \\
& \text{subject to}
& & \sum_{i: x_i \ne y_i} \langle u_{x,i}, v_{y, i} \rangle = 1 - \sum_{a \in [K]} \langle \sigma_{x,a}, \sigma_{y,a} \rangle  \mbox{ for all } x,y \in \{0,1\}^n \\
& & & \| \sigma_{x,a} \|^2 = 0 \mbox{ for all } x,a \mbox{ with } (x,a) \not \in f
\end{aligned}
\end{equation*}
\end{definition}

\begin{theorem}
\label{thm:alg}
Let $f \subseteq \{0,1\}^n \times [K]$.  Then $Q(f) = O(\ADV_{rel}^\pm(f))$.  
\end{theorem}

\begin{proof}
Let $\{\ket{\sigma_{x,a}} \}_{x \in \{0,1\}^n, a \in [K]}$ be part of an optimal solution to the program for $\ADV_{rel}^\pm(f)$.  Let $m_a$ be the dimension of $\ket{\sigma_{x,a}}$ and define 
$\ket{\psi_x} = \ket{\sigma_{x,1}} \oplus \cdots \oplus \ket{\sigma_{x,K}} \in \R^M$ where $M= \sum_{a \in [K]} m_a$.  For $a \in [K]$, let $s_a = \sum_{i < a} m_i$ and define 
$\Pi_a = \sum_{i=s_a+1}^{s_a + m_a} \ketbra{e_i}{e_i}$, where $\ket{e_i}$ is the $i^{th}$ standard basis vector.  For $A(x,y) = 1 - \braket{\psi_x}{\psi_y}$ we have 
$\gamma_2( A | D) = \ADV_{rel}^\pm(f)$.  Therefore by applying \cref{thm:lmrss} with $\epsilon = 1/3$ there is a quantum algorithm that on input $x$ 
terminates in a state $\ket{\psi_x'}$ satisfying $\langle  \psi_x' | (|\psi_x \rangle \otimes | 0 \rangle) \ge \sqrt{1-\epsilon}$ after making 
$O(\ADV_{rel}^\pm(f))$ many queries.  This implies 
$\| \ket{\psi_x'} - \ket{\psi_x}\ket{0} \|^2 \le \epsilon$ and so 
\begin{align*}
\epsilon &\ge \| \ket{\psi_x'} - \ket{\psi_x}\ket{0} \|^2 \\
&= \sum_{a \in [K]} \| (\Pi_a \otimes I) \ket{\psi_x'} - \ket{\sigma_{x,a}} \otimes \ket{0} \|^2 \\
&\ge \sum_{a: (x,a) \not \in f} \| (\Pi_a \otimes I) \ket{\psi_x'} \|^2
\end{align*}
Thus running the state generation algorithm for the states $\{\ket{\psi_x}\}$ and measuring according to the projectors $\{\Pi_a \otimes I\}_{a \in [K]}$ gives an algorithm with 
error probability at most $\epsilon = 1/3$ as desired.
\end{proof}

For the composition theorem we will also need the dual formulation of the relational adversary bound from \cref{def:adversary}.
\begin{theorem}[\cite{Bel15} Equation 22]
\label{thm:rel_adv_dual}
Let $f \subseteq \01^n \times [K]$ be a relation.  
For $i \in [n]$, let $D_i$ be a $2^n$-by-$2^n$ Boolean matrix 
where $D_i(x,y) = 1$ iff $x_i \ne y_i$.  For $a \in [K]$ let
\begin{equation}
\label{eqn:chia}
\chi_a(x) = 
\begin{cases}
1 & \mbox{ if } (x,a) \in f \\
0 & \mbox{ otherwise.}
\end{cases}
\end{equation}
Then
\begin{equation*}
\begin{aligned}
\ADV_{rel}^\pm(f) = & \underset{\Gamma}{\text{maximize}}
& & \lambda_{max}(\Gamma) \\
& \text{subject to}
& & \|\Gamma \circ D_i \| \leq 1 \mbox{ for all } i = 1, \ldots, n \\
& & & \Gamma \circ \chi_a \chi_a^T \preceq 0 \mbox{ for all } a \in [K] \enspace .
\end{aligned}
\end{equation*}
\end{theorem}

We will call a matrix $\Gamma$ satisfying $\Gamma \circ \chi_a \chi_a^T \preceq 0$ for all $a \in [K]$ a relational 
adversary matrix.

\begin{definition}[Efficiently verifiable]
\label{def:verifiable}
Let $K$ be a positive integer and $f \subseteq \{0,1\}^n \times [K]$ be a relation.  For each $a \in [K]$ define 
a Boolean function $f_a$ by $f_a(x) = 1$ iff $(x,a) \in f$.  We say that $f$ is \emph{efficiently verifiable} iff
$\ADV^\pm(f_a) = o(\ADV_{rel}^\pm(f))$ for all $a \in [K]$.  
\end{definition}

Belovs \cite{Bel15} has shown that the relational adversary bound is also a lower bound on the bounded-error quantum query complexity of 
a relation $f$ that is efficiently verifiable.

\begin{theorem}[Belovs \cite{Bel15} Theorem 40]
\label{thm:belovs}
Let $f \subseteq \01^n \times [K]$ be an efficiently verifiable relation.  Then $Q(f) = \Omega(\ADV_{rel}^\pm(f))$.  
\end{theorem}

\section{Composition Theorem}

\begin{definition}
Let $A$ be an $m$-by-$n$ matrix.  Define 
\begin{equation}
\label{eqn:hatA}
\hat A = 
\begin{bmatrix}
\|A\| I_m & A \\
A^T & \|A\| I_n
\end{bmatrix}
\enspace .
\end{equation}
\end{definition}

\begin{lemma}
\label{lem:hat_psd}
$\hat A$ is positive semidefinite for any $A$.  
\end{lemma}

\begin{proof}
The minimum eigenvalue of $A' = \begin{bmatrix} 0 & A \\ A^T & 0 \end{bmatrix}$ is $-\|A\|$ thus 
$\hat A = \|A\| I + A' \succeq 0$ by \cref{fact:lmax}.
\end{proof}

\begin{definition}[Matrix Composition]
\label{def:comp}
Let $N$ be a positive integer and let $B$ be a symmetric $2^N$-by-$2^N$ matrix.  Let $A_1, \ldots, A_N$ be matrices where $A_i$ is of size $m_i$-by-$n_i$.  
Define the \emph{matrix composition} of $B$ with $A_1, \ldots, A_N$ to be a matrix $C$ of size $\prod_{i=1}^N (m_i + n_i)$ 
with rows and columns labeled by elements of $[m_1 + n_1] \times \cdots \times [m_N + n_N]$.  For $a = (a_1, \ldots, a_N) \in [m_1 + n_1] \times \cdots \times [m_N + n_N]$ define 
$\tilde a \in \{0,1\}^N$ to be the string where $\tilde a_i = 1$ if $a_i > m_i$ and $\tilde a_i = 0$ otherwise.  Let $\tilde B$ be the matrix of the same size as $C$ where 
$\tilde B(a,b) = B(\tilde a , \tilde b)$.  Then 
\[
C = \tilde B \circ (\otimes_{i=1}^N \hat A_i) \enspace .
\]
\end{definition}

\begin{remark}
Each $A_i$ implicitly defines a function $g_i : [m_i + n_i] \rightarrow \{0,1\}$ where $g_i(a) = 1$ iff $a > m_i$.  Under this interpretation, the definition of 
$\tilde a$ arises naturally as $\tilde a = (g_1(a_1), \ldots, g_N(a_N))$ 
for $a = (a_1, \ldots, a_N) \in [m_1 + n_1] \times \cdots \times [m_N + n_N]$. 
\end{remark}

Let $C$ be the matrix composition of $B$ with $A_1, \ldots, A_N$.
The key to adversary composition theorems is the fact that $\|C\| = \|B\| \cdot \prod_{i=1}^N \|A_i\|$, as originally shown in \cite{HLS07} Lemma 16.  
The statement we give here differs in two respects from the 
statement given in \cite{HLS07}.  First, in \cite{HLS07} the lemma was only stated for $B$ restricted to be of the form 
\[
B = 
\begin{bmatrix}
0 & Z \\
Z^T & 0
\end{bmatrix}
\enspace ,
\]
for some matrix $Z$.  To show the relational composition theorem we need to allow an arbitrary symmetric matrix $B$.  Second, as the objective function of the relational adversary 
bound in \cref{thm:rel_adv_dual} is in terms of $\lambda_{max}$ we need a statement both about the spectral norm of $C$ and about $\lambda_{max}(C)$. The proof 
given in \cite{HLS07} can handle both of these generalizations; the assumption about the structure of $B$ is not used and the proof already determines the entire spectrum of $C$.  
Instead, however, we give a new proof which is substantially shorter than the one from \cite{HLS07}.
\begin{lemma}
\label{lem:main}
Let $N$ be a positive integer.  Let $B \in \R^{2^N \times 2^N}$ be a symmetric matrix and $A_1, \ldots, A_N$ be arbitrary matrices.  If $C$ is the 
matrix composition of $B$ with $A_1, \ldots, A_N$ then 
\[
\|C\| = \|B\| \cdot \prod_{i=1}^N \|A_i\|  \qquad \mbox{ and } \qquad
\lambda_{max}(C) \ge \lambda_{max}(B) \cdot \prod_{i=1}^N \|A_i\| \enspace .
\]
\end{lemma}

\begin{proof}
First we set up some notation which will be used for both the upper and lower bounds of the proof.
For $i=1, \ldots, N$ let $A_i$ be a $m_i$-by-$n_i$ matrix.  Let $\|A_i\| = \lambda^{(i)}_1 \ge \cdots \ge \lambda^{(i)}_{m_i + n_i}$ be 
the eigenvalues of the matrix $\begin{bmatrix} 0 & A_i \\ A_i^T & 0 \end{bmatrix}$ and let $z^{(i)}_1, \ldots, z^{(i)}_{m_i+n_i}$ be 
the corresponding eigenvectors.  Note that $z^{(i)}_1, \ldots, z^{(i)}_{m_i+n_i}$ form an orthonormal basis for $\R^{m_i + n_i}$.  
Finally, let $z^{(i),0}_j = z^{(i)}_j(1:m_i)$ and $z^{(i),1}_j = z^{(i)}_j(m_i+1:m_i + n_i)$, and for $b_1, b_2 \in \{0,1\}$ define $\hat A_i^{(b_1, b_2)}$ by 
\[
\begin{bmatrix}
\hat A_i^{(0,0)} & \hat A_i^{(0,1)} \\
\hat A_i^{(1,0)} & \hat A_i^{(1,1)}
\end{bmatrix}
=
\begin{bmatrix}
 \|A_i\| I_{m_i} & A_i \\
A_i^T & \|A_i\| I_{n_i}
\end{bmatrix}
\enspace .
\]

From the eigenvalue equation we see that $A_i z^{(i),1}_j = \lambda^{(i)}_j z^{(i),0}_j$ and $A_i^T z^{(i),0}_j = \lambda^{(i)}_j z^{(i),1}_j$.  This means 
\begin{equation}
\label{eq:half}
\lambda^{(i)}_j \|z^{(i),0}_j\|^2= (z^{(i),0}_j)^T A_i z^{(i),1}_j = (z^{(i),1}_j)^T A_i^T z^{(i),0}_j =  \lambda^{(i)}_j \|z^{(i),1}_j\|^2
\end{equation}
and so $\|z^{(i),0}_j\|^2 = \|z^{(i),1}_j\|^2 =1/2$ whenever $\lambda^{(i)}_j \ne 0$.  Thus for $b_1 \in \{0,1\}$ we can succinctly write 
\[
(z^{(i),b_1}_j)^T \hat A_i^{(b_1, 1-b_1)} z^{(i),1-b_1}_j = \lambda^{(i)}_j  \|z^{(i),0}_j\| \|z^{(i),1}_j\| \enspace .
\]

These observations lead to the crucial fact we need, for $b_1, b_2 \in \{0,1\}$
\begin{equation}
\label{eq:crucial}
z^{(i),b_1}_j \hat A_i^{(b_1,b_2)} z^{(i),b_2}_k= 
\begin{cases}
0 & \mbox{ if } j \ne k \\
\|A_i\| \|z^{(i),b_1}_j\|^2 & \mbox{ if } j =k , b_1 = b_2 \\
\lambda^{(i)}_j \|z^{(i),0}_j\| \|z^{(i),1}_j\|& \mbox{ if } j=k, b_1 \ne b_2
\end{cases}
\enspace .
\end{equation}

We now turn to showing the most difficult part of the lemma, that $\|C\| \le \|B\| \cdot \prod_{i=1}^N \|A_i\|$.
Let $\psi_{j_1, \ldots, j_N} = z^{(1)}_{j_1} \otimes \cdots \otimes z^{(N)}_{j_N}$ where each $j_i \in [m_i + n_i]$.  We compute
\begin{equation}
\label{eq:psi}
\psi_{j_1, \ldots, j_N}^T C \psi_{k_1, \ldots, k_N} = \sum_{\alpha, \beta \in \{0,1\}^N} B(\alpha, \beta) \prod_{i \in [N]} z^{(i),\alpha_i}_{j_i} \hat A^{(\alpha_i,\beta_i)} z^{(i),\beta_i}_{k_i} \enspace.
\end{equation}
Suppose that $j_i \ne k_i$ for some $i$.  Then by \cref{eq:crucial} $z^{(i),\alpha_i}_{j_i} \hat A^{(\alpha_i,\beta_i)} z^{(i),\beta_i}_{k_i} =0$ for all values of $\alpha_i, \beta_i$.  Thus 
$\psi_{j_1, \ldots, j_N}^T C \psi_{k_1, \ldots, k_N} = 0$ unless $j_i = k_i$ for all $i = 1, \ldots, N$.  

Let us now consider the value in this case.
For $i = 1, \ldots, N$ let 
\[
D_i = \begin{bmatrix} 
\|A_i\| & \lambda^{(i)}_{j_i}  \\ 
\lambda^{(i)}_{j_i}  & \|A_i\| 
\end{bmatrix}
\circ 
\left(
\begin{bmatrix} 
\|z^{(i),0}_j\| \\ 
\|z^{(i),1}_j\|
\end{bmatrix}
\begin{bmatrix} 
\|z^{(i),0}_j\| & \|z^{(i),1}_j\|
\end{bmatrix}
\right) \enspace .
\]
Each $D_i \succeq 0$ because it is the Hadamard product of two positive semidefinite matrices (the first matrix is diagonally dominant as $\|A_i\| \ge | \lambda^{(i)}_{j_i} |$ and 
so is psd). Thus $\|D_i\|_{tr} = \Tr(D_i) = \|A_i\|$.  Let $D = \otimes_{i=1}^N D_i$ and note that $\|D\|_{tr} = \prod_{i=1}^N \|A_i\|$.  Then from \cref{eq:crucial} and \cref{eq:psi} we see that 
\[
\psi_{j_1, \ldots, j_N} ^T C \psi_{j_1, \ldots, j_N} = \Tr(BD) \le \|B\| \|D\|_{tr} = \|B\| \prod_{i=1}^N \|A_i\| \enspace .
\]

Let $M = \prod_{i=1}^N (m_i + n_i)$ and let $\psi = \sum_{j_1, \ldots, j_N} \alpha_{j_1, \ldots, j_N} \psi_{j_1, \ldots, j_N} \in \R^M$ be an arbitrary unit vector.  Then 
\begin{align*}
\psi^T C \psi &= \sum_{j_1, \ldots, j_N,k_1, \ldots, k_N} \alpha_{j_1, \ldots, j_N} \alpha_{k_1, \ldots, k_N} (\psi_{j_1, \ldots, j_N}^T C \psi_{k_1, \ldots, k_N}) \\
&= \sum_{j_1, \ldots, j_N} \alpha_{j_1, \ldots, j_N}^2 (\psi_{j_1, \ldots, j_N}^T C \psi_{j_1, \ldots, j_N}) \\
&\le \|B\| \prod_{i=1}^N \|A_i\| \enspace .
\end{align*}
This shows $\|C\| \le \|B\| \cdot \prod_{i=1}^N \|A_i\|$.

We now turn to showing the lower bounds 
\begin{align}
\label{eq:sp_lower}
\|C\| &\ge \|B\| \cdot \prod_{i=1}^N \|A_i\| \\
\label{eq:lmax_lower}
\lambda_{max}(C) &\ge \lambda_{max}(B) \cdot \prod_{i=1}^N \|A_i\| \enspace .
\end{align}
Recall that $z^{(i)}_1$ is an eigenvector of $\begin{bmatrix} 0 & A_i \\ A_i^T & 0 \end{bmatrix}$ corresponding 
to eigenvalue $\|A_i\|$.  Also $\| z^{(i),0}_1\|^2 = \| z^{(i),1}_1\|^2 = 1/2$ by \cref{eq:half}.    Let $v$ be a unit norm eigenvector of $B$ corresponding to 
eigenvalue $\lambda$ (which later will either be set to $\lambda_{max}(B)$ or $\lambda_{min}(B)$).  For $x = (x_1, \ldots, x_N) \in \{0,1\}^{Nm}$ let $\tilde x = (g(x_1), \ldots, g(x_N))$.  We now define our witness $\psi$ which we will 
use to show \cref{eq:sp_lower,eq:lmax_lower} via \cref{fact:lmax}:
\[
\psi((x_1, \ldots, x_N)) = v(\tilde x) \prod_{i=1}^N \sqrt{2} z^{(i),\tilde x_i}_1(x_i) \enspace.
\]
We have 
\begin{align*}
\sum_{x_1, \ldots, x_N} \psi((x_1, \ldots, x_N))^2 &= \sum_{\alpha \in \{0,1\}^N} v(\alpha)^2 \prod_{i=1}^N \left( \sum_{x_i : \tilde x_i = \alpha_i} 2 z^{(i),\alpha_i}_1(x_i)^2 \right) \\
&=  \sum_{\alpha \in \{0,1\}^N} v(\alpha)^2 \prod_{i=1}^N  \left( 2 \|z^{(i),\alpha_i}_1\|^2 \right) \\
&= \sum_{\alpha \in \{0,1\}^N} v(\alpha)^2 \\
&= 1 \enspace .
\end{align*}
Thus $\psi$ is a unit vector. 
Now
\begin{align*}
\psi^T C \psi &= \sum_{\alpha, \beta \in \{0,1\}^N} B(a,b) v(\alpha) v(\beta) \prod_{i=1}^N 2 \cdot \left( z^{(i),\alpha_i}_{1} \hat A^{(\alpha_i,\beta_i)} z^{(i),\beta_i}_{1} \right) \\
& = \sum_{\alpha, \beta \in \{0,1\}^N} B(a,b) v(\alpha) v(\beta) \prod_{i=1}^N \|A_i\| \\
& = \lambda \prod_{i=1}^N \|A_i\| \enspace .
\end{align*}
Taking $\lambda = \lambda_{max}(B)$ shows \cref{eq:lmax_lower} by \cref{fact:lmax}.  If $\lambda_{max}(B) = \|B\|$ this also 
shows \cref{eq:sp_lower}. Otherwise, if $\|B\| = -\lambda_{min}(B)$ then taking $\lambda = \lambda_{min}(B)$ we have $\psi^T C \psi = - \|B\| \prod_{i=1}^N \|A_i\|$, 
implying \cref{eq:sp_lower} again by \cref{fact:lmax}.
\end{proof}

We now turn to showing the main result of this note, that $\ADV_{rel}^\pm(f \circ g^N) = \ADV_{rel}^\pm(f) \ADV^\pm(g)$ for a relation $f$ and Boolean function $g$.  
We start with the more difficult direction which is showing the lower bound.
\begin{theorem}
\label{thm:lower}
Let $f \subseteq \01^N \times [K]$ be a relation and $g: \01^m \rightarrow \01$ be a Boolean function.  Then $\ADV_{rel}^\pm(f \circ g^N) \ge \ADV_{rel}^\pm(f) \ADV^\pm(g)$.  
\end{theorem}

\begin{proof}
Let $\Gamma_f$ be a relational adversary matrix achieving the optimal bound for the relation $f$.
Let
\[
\Gamma_g = 
\begin{bmatrix}
0 & Z \\
Z^T & 0
\end{bmatrix}
\]
be an optimal functional adversary matrix for $g$.  Define $\Gamma_h$ 
as the matrix composition of $\Gamma_f$ with $N$ copies of $Z$.  

For $a \in [K]$ let $\chi_a \in \R^{2^N}$ be defined as in \cref{eqn:chia} and define $\phi_a \in \R^{2^{mN}}$ similarly but for the composed function:
\[
\phi_a((x_1, \ldots, x_N))
=
\begin{cases}
1,& \text{$((g(x_1), \ldots, g(x_N)),a) \in f$;}\\
0,& \text{otherwise.}
\end{cases}
\]

We will show three things 
\begin{enumerate}[(1)]
  \item $\lambda_{max}(\Gamma_h ) =\lambda_{max}(\Gamma_f) \cdot \|\Gamma_g\|^N$;
  \item $\Gamma_h \circ \phi_a \phi_a^T \preceq 0$ for all $a \in [K]$; and
  \item $\|\Gamma_h \circ D_\ell \| \le \|\Gamma_f \circ D_p\| \cdot \|\Gamma_g \circ D_q\| \cdot \|\Gamma_g\|^{N-1}$, for any $p \in [N], q \in [m]$ where $\ell=(p-1)m+q$ is the $q^{th}$ bit in the $p^{th}$ block.
\end{enumerate}
These three items give the theorem since item~(2) implies $\Gamma_h$ is a relational adversary matrix and 
\begin{align*}
\frac{\lambda_{max}(\Gamma_h)}{\| \Gamma_h \circ D_\ell \|} & \ge \frac{\lambda_{max}(\Gamma_f) \cdot \|\Gamma_g\|^N}{\| \Gamma_f \circ D_p\| \cdot \|\Gamma_g \circ D_q\| \cdot \|\Gamma_g\|^{N-1} } \\
&\ge \left( \frac{\lambda_{max}(\Gamma_f)}{\|\Gamma_f \circ D_p\|}\right) \left(\frac{\|\Gamma_g\|}{\|\Gamma_g \circ D_q\|} \right) \\
&\ge \ADV_{rel}^\pm(f) \ADV^{\pm}(g) \enspace .
\end{align*}

We now show the three items.
Item~(1) follows immediately from the second equation of \cref{lem:main} as $\Gamma_h$ is defined to be the matrix composition of $\Gamma_f$ with $N$ copies 
of $Z$, and $\|\Gamma_g\| = \| Z \|$.
\medskip

Item~(2) is the main novelty in the relational case.  From \cref{def:comp} we see that 
\[
\Gamma_h \circ \phi_a \phi_a^T = (\tilde \Gamma_f \circ \phi_a \phi_a^T) \circ (\otimes^N \hat \Gamma_g) \enspace ,
\]
where $\hat\Gamma_g$ is defined via \cref{eqn:hatA} and $\tilde \Gamma_f$ is a $2^{mN}$-by-$2^{mN}$ matrix defined as 
\[
\tilde \Gamma_f((x_1, \ldots, x_N), (y_1, \ldots, y_N)) = \Gamma_f( (g(x_1), \ldots, g(x_N)), (g(y_1), \ldots, g(y_N)) \enspace .
\]
By Lemma~\ref{lem:hat_psd}, $\hat \Gamma_g \succeq 0$, which gives us $\otimes^N \hat \Gamma_g \succeq 0$.
Also, $\Gamma_f\circ \chi_a \chi_a^T \preceq 0$ by the definition of the relational adversary matrix, thus by Fact~\ref{fact:duplication}, we get that $\tilde \Gamma_f \circ \phi_a \phi_a^T\preceq 0$.
Combining the last two observations with Fact~\ref{fact:simple}, we get Item~(2).
\medskip

Item~(3) follows as in \cite{HLS07}, but for completeness we give the details here.  Write $\ell \in [mN]$ as 
$\ell = (p-1)m + q$ with $p \in [N]$ and $q \in [m]$, i.e.\ $\ell$ refers to $x_p(q)$, the $q^{th}$ bit in the $p^{th}$ block. 
It is shown in \cite{HLS07} that 
\begin{equation}
\label{eq:last_bit}
\Gamma_h \circ D_\ell = (\widetilde{\Gamma_f \circ D_p}) \circ \left( (\otimes^{p-1}\; \hat{\Gamma}_g) \otimes (\hat{\Gamma}_g \circ D_q) \otimes (\otimes^{N-p}\; \hat{\Gamma}_g) \right) \enspace ,
\end{equation}
where
\[
(\widetilde{\Gamma_f \circ D_p}) ((x_1, \ldots, x_N), (y_1, \ldots, y_N)) = (\Gamma_f \circ D_p)( (g(x_1), \ldots, g(x_N)), (g(y_1), \ldots, g(y_N)) \enspace .
\]
The right hand side is not obviously in the proper form to apply \cref{lem:main} because \cref{lem:main} requires $\widehat{\Gamma_g \circ D_q}$ in place of the term $\hat{\Gamma}_g \circ D_q$ in \cref{eq:last_bit}.  
We can show, however, that the right hand side of \cref{eq:last_bit} does not change if we replace $\hat{\Gamma}_g \circ D_q$ with $\widehat{\Gamma_g \circ D_q}$.
\begin{claim}
\label{clm:right_form}
\[
(\widetilde{\Gamma_f \circ D_p}) \circ \left( (\otimes^{p-1}\; \hat{\Gamma}_g) \otimes (\hat{\Gamma}_g \circ D_q) \otimes (\otimes^{N-p}\; \hat{\Gamma}_g) \right) =
(\widetilde{\Gamma_f \circ D_p}) \circ \left( (\otimes^{p-1}\; \hat{\Gamma}_g) \otimes (\widehat{\Gamma_g \circ D_q}) \otimes (\otimes^{N-p}\; \hat{\Gamma}_g) \right)
\]
\end{claim}
\begin{proof}
$\hat{\Gamma}_g \circ D_q$ and $\widehat{\Gamma_g \circ D_q}$ only differ on the diagonal, therefore it suffices to show that for all entries $(x_1, \ldots, x_N), (y_1, \ldots, y_N)$ 
where $x_p = y_p$ the left and right hand sides agree.  When $x_p = y_p$ the left hand side will be $0$ because $\hat{\Gamma}_g \circ D_q$ is zero on the diagonal.  When $x_p = y_p$ 
the right hand side will also be zero because then $\tilde x_p = \tilde y_p$ and so $(\widetilde{\Gamma_f \circ D_p})(x,y) = 0$.  
\end{proof}
The right hand side of the expression in \cref{clm:right_form} is now in the right form to apply the first equation of \cref{lem:main}, and we can conclude
$\|\Gamma_h \circ D_\ell\| \le \|\Gamma_f \circ D_p\| \|\Gamma_g \circ D_q\| \|\Gamma_g\|^{N-1}$ as desired.
\end{proof}

Next we show the upper bound.
\begin{theorem}
\label{thm:upper}
Let $f \subseteq \01^N \times [K]$ be a relation and $g: \01^m \rightarrow \01$ be a Boolean function.  Then $\ADV_{rel}^\pm(f \circ g^N) \le \ADV_{rel}^\pm(f) \ADV^\pm(g)$.  
\end{theorem}

\begin{proof}
Let $\{u_{y,i}\}_{y \in \{0,1\}^m, i \in [m]},\{u_{y,i}\}_{y \in \{0,1\}^m, i \in [m]}$ be an optimal solution to the $\ADV^\pm(g)$ program from \cref{thm:adv_dual}, and let 
$\{\psi_{z,i}\}_{z \in \{0,1\}^N, i \in [N]},\{\phi_{z,i}\}_{z \in \{0,1\}^N, i \in [N]}, \{\sigma_{z,a}\}_{z \in \{0,1\}^N, a \in [K]}$ be an optimal solution to the $\ADV_{rel}^\pm(f)$ program 
from \cref{def:adv_rel}.  We will construct a solution to the program from \cref{def:adv_rel} for $h=f \circ g^N$ of cost $\ADV_{rel}^\pm(f) \ADV^\pm(g)$.  For 
$x = (x_1, \ldots, x_N) \in \{0,1\}^{Nm}$ let $\tilde x = (g(x_1), \ldots, g(x_N))$.

Define 
\begin{itemize}
  \item $\alpha_{x,\ell} = \psi_{\tilde x, p} \otimes u_{x_p, q}$ for $\ell = (m-1)p + q$ where $p \in [N], q \in [m]$.  
   \item $\beta_{x,\ell} = \phi_{\tilde x, p} \otimes v_{x_p, q}$ for $\ell = (m-1)p + q$ where $p \in [N], q \in [m]$.
   \item $\rho_{x,a} = \sigma_{\tilde x, a}$.
\end{itemize}
Then for any $x = (x_1, \ldots, x_N) \in \{0,1\}^{Nm}$
\begin{align*}
\sum_\ell \| \alpha_{x,\ell} \|^2 &= \sum_p \| \psi_{\tilde x, p} \|^2 \sum_{i \in [m]} \| u_{x_p,i} \|^2 \\
&\le \ADV^\pm(g) \sum_p \| \psi_{\tilde x, p} \|^2 \\
&\le \ADV_{rel}^\pm(f) \ADV^\pm(g) \enspace .
\end{align*}
A similar calculation shows $\sum_\ell \| \beta_{x,\ell} \|^2 \le \ADV_{rel}^\pm(f) \ADV^\pm(g)$ for any $x \in \{0,1\}^{Nm}$.  

Having established the objective value, we move on to the constraints.
\begin{align*}
\sum_{\ell: x_\ell \ne y_\ell} \langle \alpha_{x,\ell}, \beta_{x, \ell} \rangle &= \sum_{p \in [N]} \braket{\psi_{\tilde x, p}}{\phi_{\tilde y, p}} \sum_{i \in [m] \atop x_p(i) \ne y_p(i)} \braket{u_{x_p,i}}{v_{y_p,i}} \\
&= \sum_{p \in [N] \atop \tilde x_p  \ne \tilde y_p} \braket{\psi_{\tilde x, p}}{\phi_{\tilde y, p}} \\
&= 1 - \sum_{a: (\tilde x, a) \in f, (\tilde y, a) \in f} \braket{\sigma_{\tilde x,a}}{\sigma_{\tilde y,a}} \\
&= 1 - \sum_{a: (x, a) \in h, (y, a) \in h} \braket{\rho_{x,a}}{\rho_{y,a}} \enspace ,
\end{align*}
as desired.

Finally,  $\| \rho_{x,a} \|^2 = \| \sigma_{\tilde x ,a}\|^2 = 0$ if $(\tilde x, a) \not \in f$.  As $(x,a) \in h$ iff $(\tilde x,a) \in f$ this shows $\| \rho_{x,a} \|^2 = 0$ for all 
$(x,a) \not \in h$.
\end{proof}

\begin{corollary}
Let $f \subseteq \{0,1\}^N \times [K]$ be a relation and $g: \{0,1\}^m \rightarrow \{0,1\}$ be a Boolean function.  Then $\ADV_{rel}^\pm(f \circ g^N) = \ADV_{rel}^\pm(f) \ADV^\pm(g)$.  
\end{corollary}

\begin{proof}
This follows from \cref{thm:lower} and \cref{thm:upper}.
\end{proof}

\begin{corollary}
\label{cor:pcomp}
Let $K$ be a positive integer, $f \subseteq \{0,1\}^N \times [K]$ be an efficiently verifiable relation, and $g: \{0,1\}^m \rightarrow \{0,1\}$ be a Boolean function.  
Then $Q(f \circ g^N) = \Theta(\ADV_{rel}^\pm(f) \ADV^\pm(g))$.
\end{corollary}

\begin{proof}
Let $h = f \circ g^N$.
We start by showing $Q(h) = \Omega(\ADV_{rel}^\pm(f) \ADV^\pm(g))$.
By \cref{thm:lower} $\ADV_{rel}^\pm(h) \ge \ADV_{rel}^\pm(f) \ADV^\pm(g)$.  For $a \in [K]$ let $h_a: \{0,1\}^{mN} \rightarrow \{0,1\}$ be defined as 
$h_a((x_1, \ldots, x_N)) = 1$ iff $((g(x_1), \ldots, g(x_N)),a) \in f$.  Letting $f_a : \{0,1\}^N \rightarrow \{0,1\}$ be defined as $f_a(x) = 1$ iff $(x,a) \in f$, we 
see that $h_a = f_a \circ g^N$.  By \cref{thm:basic} we have that $\ADV^\pm(h_a) = \ADV^\pm(f_a) \ADV^\pm(g)$ and therefore, as $f$ is efficiently 
verifiable, $\ADV^\pm(h_a) = o(\ADV_{rel}^\pm(h))$ for every $a \in [K]$.  Thus $h$ is also efficiently verifiable and the corollary follows by \cref{thm:belovs}.  

For the other direction, we use \cref{thm:upper} to obtain $\ADV_{rel}^\pm(h) \le \ADV_{rel}^\pm(f) \ADV^\pm(g)$ and 
then apply \cref{thm:alg}.
\end{proof}

\begin{corollary}
Let $K$ be a positive integer, $f \subseteq \{0,1\}^N \times [K]$ be an efficiently verifiable relation, and $g: \{0,1\}^m \rightarrow \{0,1\}$ be a Boolean function.  
Then $Q(f \circ g^N) = \Theta(Q(f) Q(g))$.
\end{corollary}

\begin{proof}
By \cref{cor:pcomp}, $Q(f \circ g^N) = \Theta(\ADV_{rel}^\pm(f) \ADV^\pm(g))$.  We have $Q(g) = \Theta(\ADV^\pm(g))$ by \cref{thm:adv_char}.  Also,
$Q(f) = O(\ADV_{rel}^\pm(f))$ by \cref{thm:alg} and $Q(f) = \Omega(\ADV_{rel}^\pm(f))$ by \cref{thm:belovs} as $f$ is efficiently verifiable.  The corollary follows.
\end{proof}

\section*{Acknowledgements}
We thank Ronald de Wolf for asking us if an adversary composition theorem holds for relations, the question that began this note, and for helpful comments on 
the manuscript. 
A.B. is supported by the ERDF project number 1.1.1.2/I/16/113.
T.L. is supported by the Australian Research Council Grant No: DP200100950.


\newcommand{\etalchar}[1]{$^{#1}$}

\end{document}